\documentclass[sigconf]{acmart}
\usepackage{graphicx}
\AtBeginDocument{%
  }

\setcopyright{acmlicensed}
\copyrightyear{2026}
\acmYear{2026}
\acmDOI{XXXXXXX.XXXXXXX}
\acmConference[Conference acronym 'XX]{Make sure to enter the correct
  conference title from your rights confirmation email}{June 03--05,
  2026}{Woodstock, NY}
\acmISBN{978-1-4503-XXXX-X/2026/06}

\usepackage{newverbs}
\usepackage{fancyvrb}
\usepackage{cprotect}
\usepackage{placeins}
\usepackage{float}
\usepackage{xcolor}

\newcommand{\TTJ}{\mathsf{TTJ}}

\newcommand{\SAAD}{\mathsf{CODA}}
\newcommand{\Query}{Q}
\newcommand{\RIP}{\mathsf{RIP}}
\newcommand{\Repls}[1]{\operatorname{Repl}^*(#1)}
\newcommand{\Repl}[2]{\operatorname{Repl}(#1,#2)}
\newcommand{\Virt}[1]{\operatorname{Virt}(#1)}
\newcommand{\Par}[2]{\operatorname{Par}(#1, #2)}
\newcommand{\Child}[2]{\operatorname{Child}(#1, #2)}
\newcommand{\Root}[1]{\operatorname{Root}(#1)}
\newcommand{\Merge}[2]{\operatorname{Mrg}(#1,#2)}
\newcommand{\Choose}[1]{\operatorname{Choose}(#1)}

\newcommand{\join}{\Join^T}
\newcommand{\concat}{\mathbin{+\mkern-5mu+}}

\newcommand{\matchingtuples}{\mathsf{MT}}

\newcommand{\keys}{\textsf{keys}}
\newcommand*{\myverb}[1]{\protect\Verb!#1!\xspace}
\newcommand*{\Router}{\protect\Verb!Router!\xspace}

\newcommand*{\Rinner}{\protect\Verb!Rinner!\xspace}

\newcommand*{\Row}{\protect\Verb!r!\xspace}

\newcommand*{\None}{\protect\Verb!None!\xspace}
\newcommand*{\EmptyList}{\myverb{[]}}
\newcommand*{\Equal}{\myverb{=}}
\newcommand*{\aMatchingTuple}{\myverb{m}}
\newcommand{\Leaf}{\ell}
\newcommand{\Set}[1]{\{#1\}}
\newcommand{\SetOf}[2]{\{#1 \mid #2\}}
\newcommand{\Attr}[1]{\Sigma(#1)}
\newcommand*{\MaxRelSize}{\texttt{max\_rel\_size}\xspace}

\newcommand*{\MaxQSize}{\texttt{max\_size}\xspace}

\usepackage[utf8]{inputenc}
\usepackage{relsize}
\usepackage{pifont}
\newcommand{\bigsymbol}[1]{\raisebox{-0.8pt}{\textlarger{#1}}}
\newcommand{\circleone}{\bigsymbol{\ding{172}}}
\newcommand{\circletwo}{\bigsymbol{\ding{173}}}
\newcommand{\circlethree}{\bigsymbol{\ding{174}}}
\newcommand{\circlefour}{\bigsymbol{\ding{175}}}
\newcommand{\circlefive}{\bigsymbol{\ding{176}}}
\newcommand{\circlesix}{\bigsymbol{\ding{177}}}

\usepackage{cleveref}
\crefname{algocf}{Algorithm}{Algorithms}
\Crefname{algocf}{Algorithm}{Algorithms}
\AtBeginDocument{}
\crefname{figure}{Figure}{Figures}
\Crefname{figure}{Figure}{Figures}
\crefname{table}{Table}{Tables}
\Crefname{table}{Table}{Tables}
\crefname{section}{Section}{Sections}
\Crefname{section}{Section}{Sections}
\crefname{algorithm}{Algorithm}{Algorithms}
\Crefname{algorithm}{Algorithm}{Algorithms}

\usepackage{listings}
\usepackage{xcolor}
\usepackage[linesnumbered,ruled,vlined, nofillcomment, algosection]{algorithm2e}

\SetCommentSty{mycommfont}

\SetKwInput{KwInput}{Input}                
\SetKwInput{KwOutput}{Output}          
\SetKwInput{KwGlobalVariable}{Instance variables}
\SetKwInput{KwGlobalVariables}{Instance variables}
\SetKwInput{KwPurpose}{Purpose}
\SetKw{Continue}{continue}

\begin{document}

\title[Computer-Orchestrated Design of
  Algorithms]{Computer-Orchestrated Design of Algorithms: From Join
  Specification to Implementation}

\author{Zeyuan Hu}
\affiliation{%
  \institution{University of California, Los Angeles}
  \city{Los Angeles}
  \state{California}
  \country{USA}
}
\email{zeyuanhu@cs.ucla.edu}
\orcid{0000-0003-3036-2777}

\begin{abstract}
Equipping query processing systems with provable theoretical
guarantees has been a central focus at the intersection of database
theory and systems in recent years. However, the inevitable divergence
between theoretical abstractions and system assumptions creates a gap
between an algorithm's high-level logical specification and its
low-level physical implementation. Ensuring the correctness of this
logical-to-physical translation is crucial for realizing theoretical
optimality as practical performance gains. Existing database testing
frameworks struggle to address this need because necessary
algorithm-specific inputs such as join trees (as opposed to query
plans) are absent from standard test case generation, and integrating
complex algorithms into these frameworks imposes prohibitive
engineering overhead. Fallback solutions, such as using
macro-benchmark queries, are inherently too noisy for isolating
intricate defects during this translation.

In this experience paper, we present a retrospective analysis of
$\SAAD$, a computer-orchestrated testing framework utilized during the
logical-physical co-design of TreeTracker Join ($\TTJ$), a theoretically
optimal yet practical join algorithm recently published in ACM
TODS. By synthesizing minimal reproducible examples, $\SAAD$
successfully isolates subtle translation defects, such as state
mismanagement and intricate mapping conflicts between join trees and
bushy plans. We demonstrate that this logical-to-physical translation
process is a bidirectional feedback loop: early structural testing not
only hardened $\TTJ$'s physical implementation but also exposed a
critical boundary condition that directly refined the formal
precondition of $\TTJ$ itself. Finally, we detail how confronting
these structural translation challenges drove the architectural
evolution of $\SAAD$ into a robust, structure-aware test generation
pipeline for join-tree-dependent algorithms.
\end{abstract}





\maketitle

\section{Introduction}
Query processing remains the cornerstone of any relational database
management system (DBMS). In recent years, a major research thrust has
emerged to equip practical query evaluation with provable theoretical
guarantees. At the heart of this effort is the evaluation of
conjunctive queries (CQs), queries involving multiple joins. An
important subclass of CQs is acyclic conjunctive queries (ACQs) due to
their immense abundance in real-life
workloads~\cite{Luo2026,Fischl2021}. This line of research aims to
bridge the gap between theory and practice in ACQ evaluation by
introducing algorithms that are efficient in both domains (see
survey~\cite{Koutris2026}). Notably, this investigation is no longer
merely an academic exercise; such theoretical advancements for ACQs
have proven their practical viability in production-grade systems,
being successfully deployed in DuckDB~\cite{Qiao2026} and
independently realized within Microsoft SQL Server's optimization
infrastructure~\cite{Zhao2026}. The latter shows a promising
direction: re-evaluating mature query optimization heuristics through
formal lenses, such as ACQ evaluation. A major challenge in designing
practical ACQ evaluation algorithms is that system constraints
commonly diverge from theoretical assumptions. Hence, a gap exists
between the high-level logical specification of an algorithm and its
low-level physical implementation. Ensuring a faithful
logical-to-physical translation is therefore crucial to guarantee that
theoretical optimality actually translates into practical performance
gains.

A natural choice to verify the physical translation of such algorithms
is to leverage conventional database testing
frameworks~\cite{Tang2023,Rigger2020,Zhang2025,Lai2025,Seltenreich,Jung2019,Song2023,Jiang2024,Zhong2020}. However,
incompatibilities exist with the needs of the translation. First,
algorithms rely on an input data structure called a join
tree~\footnote{It is an unfortunate coincidence that the term ``join
tree'' from database theory~\cite{Abiteboul1995} may also refer to a
query plan in DBMS literature. In this paper, we make a distinction
between join tree and query plan.} that is distinct from a query plan
and is not commonly seen in mainstream databases. This fact imposes a
significant challenge to conventional database testing frameworks, as
the correctness of query evaluation depends not only on the query
plans but also on the topology of join trees. Unfortunately, the
existing testing frameworks do not support join trees, and
retrofitting them is highly non-trivial. Second, existing frameworks
generate SQL queries and database instances with the goal of testing
fully-featured DBMSs. Thus, to utilize these tools, researchers would
first have to fully integrate their nascent algorithm into an existing
DBMS, which imposes an unnecessary engineering burden when the primary
goal is simply to verify the correctness of the logical-to-physical
translation. Hence, it is natural for researchers to use
macro-benchmark queries as a proxy to test their physical translation
of an algorithm specification on an experimental prototype that
closely mimics system constraints. They have to manually craft a join
tree for a given query and check the correctness of the translation
through evaluation. The issue with this approach is due to the
complexity of macro-benchmarks. Since the purpose of macro-benchmark
queries is to mimic real-life workloads and benchmark DBMS
performance, those queries, along with the associated database
instance, are inherently complex. When a defect is flagged, the sheer
size and structural noise of the macro-benchmark query prevent
researchers from effectively isolating the root cause.

\begin{figure*}
  \centering
  \includegraphics[width=\textwidth]{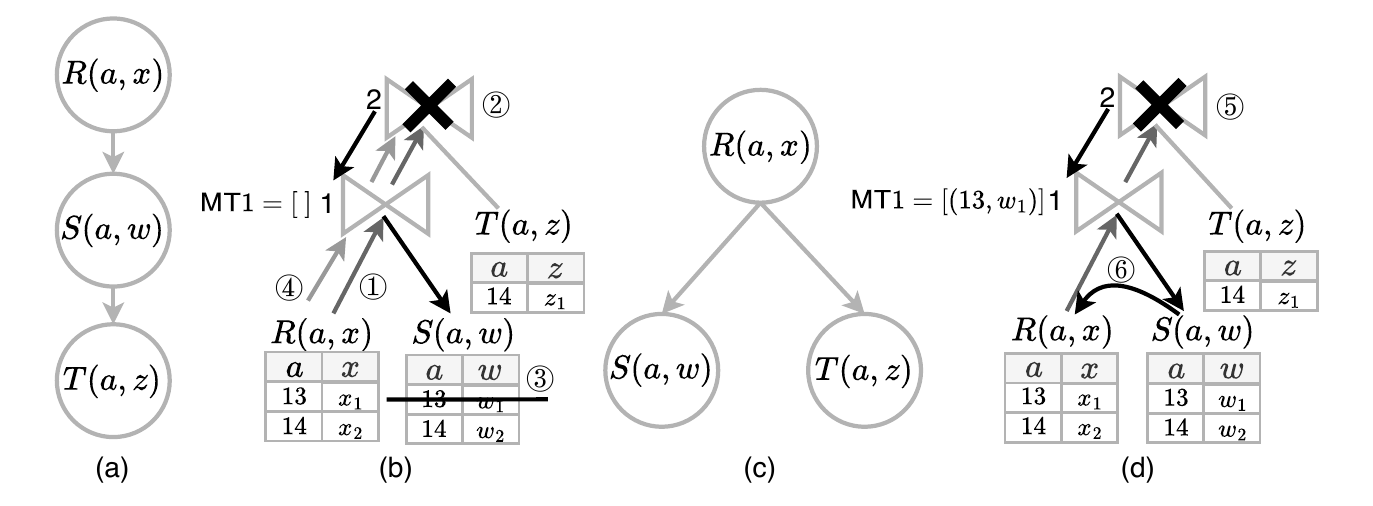}
  \cprotect\caption{Illustration of how the evaluation correctness of
    $\Query = R(a,x) \Join S(a,w) \Join T(a,z)$ depends on the join
    tree structure. Using a defective $\TTJ$ iterator (missing
    Line~\ref{tt-1-beta-b-hashtable-tq:matchingtuples-nil-pass-context}
    in \cref{algo:ttj-join}), $\Query$ evaluates correctly under the
    join tree in (a), as traced in (b). Conversely, under the join
    tree in (c), the identical defective iterator yields an incorrect
    empty result, as traced in (d). We explain the
      details in \cref{sec:motivating-example}. Symbols \circleone
      through \circlesix mark execution snapshots referenced in the
      text.}
  \label{fig:join-trees}
\end{figure*}

In this paper, we detail the Computer-Orchestrated Design of
Algorithms ($\SAAD$)\footnote{In music, a coda is a concluding passage
added to the end of a musical piece, providing a sense of finality,
closure, or resolution. Similarly, we believe the design of any
practical algorithm is truly concluded only when it is rigorously
verified and successfully deployed in practice.} framework utilized
during the design of TreeTracker Join ($\TTJ$), an optimal yet
practical join algorithm recently published in ACM Transactions on
Database Systems (TODS)~\cite{Hu2025}. The term
``computer-orchestrated'' emphasizes the framework's
researcher-centric perspective: instead of black-box bug discovery,
$\SAAD$ explicitly orchestrates test generation based on structural
properties of join trees to assist researchers in verifying the
intricate logical-to-physical translation. $\SAAD$ operates directly
on join trees, verifying the correctness of the physical
implementation by exploiting the interaction between join trees and
query plans. Furthermore, $\SAAD$ proactively exposes defects via
\emph{minimal reproducible examples} (MREs), allowing us to rapidly
identify subtle translation errors easily masked by macro-benchmark
proxies. Notably, the core components of $\SAAD$ are immediately
applicable to a broader class of join-tree-dependent
algorithms~\cite{Koutris2026}.

Importantly, this retrospective analysis reveals that orchestrating
algorithm design creates a bidirectional feedback loop: the defects
exposed by $\SAAD$ did not merely debug $\TTJ$'s physical
implementation; in one instance, they refined the theoretical
precondition of $\TTJ$ itself. Simultaneously, tackling these complex
translation challenges drove the evolution of $\SAAD$ itself,
culminating in a robust, structure-aware testing framework.

\section{From Logical to Physical}
\label{sec:logical-physical}

$\TTJ$ evaluates an ACQ such as $Q = R(a,x) \Join S(a,w) \Join T(a,z)$
for relations $R$, $S$, and $T$ optimally. While its logical
specification is a simple 13-line recursive function relying on
implicit call-stack state (\cref{fig:ttj}), to ensure an
apples-to-apples performance comparison against mature systems like
PostgreSQL, $\TTJ$ must be implemented as a physical iterator within a
standard pipelined execution engine (\cref{algo:ttj-join}, adapted
from~\cite{Hu2024}). This logical-to-physical translation is highly
non-trivial: it demands explicit mutual recursion, state tracking, and
complex backtracking, essentially doubling the length of its recursive
counterpart. The key idea of $\TTJ$ is to leverage a data structure
called a join tree (different from a query plan; we will explain in
\cref{sec:motivating-example}) to avoid redundant computation and
guarantee optimality. Due to space constraints, we defer
step-by-step query evaluation examples of \cref{fig:ttj}
to~\cite{Hu2025}, and the full formal proofs of $\TTJ$ iterator's
properties to~\cite{Hu2024}.

\begin{figure}
\begin{lstlisting}[escapeinside={(*}{*)}]
def ttj(t, plan, i):
  if i > plan.len(): print(t)
  else: 
    R = plan[i]; k = $\pi_{\keys(\text{plan[1..i],R})}(t)$
    P = parent(plan[1..i], R)(*\label{lst:ttj-parent}*)
    if R[k] is None & P is not None:(*\label{lst:ttj-lookup}*)
      return P # backjump to P (*\label{lst:backjump}*)
    for r in R[k]: (*\label{lst:ttj-loop}*)
      result = ttj(t$\concat$r, plan, i+1)(*\label{lst:ttj-recursive}*)
      if result == R: # catch backjump (*\label{lst:ttj-catch}*)
        R[k].delete(r)(*\label{lst:ttj-delete}*)
      elif result is not None:
        return result # continue backjump
\end{lstlisting}
\caption{Logical $\TTJ$ specification from \cite{Hu2025}, featuring
  the simplicity of implicit state management via the call stack.}
\label{fig:ttj}
\end{figure}%

\begin{algorithm}
	\caption{Physical $\TTJ$ implementation featuring explicit stack
          and state management alongside mutual recursion between
          \myverb{getNext()} and \myverb{deleteDT(R)}.}
	\label{algo:ttj-join}
	\DontPrintSemicolon
	
	\KwPurpose{An iterator returns, one at a time, the join result of \Router and \Rinner.}
	\KwOutput{A tuple from \Router $\Join$ \Rinner}
	
	\SetKwFunction{FOpen}{open}
	\SetKwFunction{FGetNext}{getNext}
	\SetKwFunction{FNext}{next}
	\SetKwFunction{FClose}{close}
	\SetKwFunction{FPassContext}{deleteDT}
	\SetKwFunction{FGetROne}{GetOuter}
	\SetKwFunction{FGetRTwo}{GetR2}
	\SetKwFunction{FLookUpH}{LookUpH}
	\SetKwFunction{FReset}{Reset}
	\SetKwFunction{FGet}{get}
	
        \SetKw{is}{is}
        \SetKw{void}{void}
        \SetKw{isnot}{is not}
        \SetKw{and}{and}
	
	\SetKwProg{Fn}{}{}{}
	\SetKwFunction{FOpenPrefix}{\void open}%
	\SetKwFunction{FGetNextPrefix}{Tuple getNext}
	\SetKwFunction{FPassContextPrefix}{Tuple deleteDT}
	\SetKwFunction{FTTJClass}{TTJIterator}
	
	\Fn{\FTTJClass}
	{		
		\Fn{\FOpenPrefix{}}
		{
			Initialize \Row, $\matchingtuples$ to \None
			
			\Rinner.\FOpen{}
                        
			Build hash table \myverb{inner} from \Rinner
                        
			\Router.\FOpen{}
		}
		
		\Fn{\FGetNextPrefix{}}
		{
			\If {$\matchingtuples$ \isnot{} \None \and
                          $\matchingtuples$ \isnot{} \EmptyList} {
				\label{new-ttj-v2:SetOfMatchingTuples-non-nill-non-empty}
				
                                \If {\aMatchingTuple \Equal
                                  $\matchingtuples$.\FNext{} \and
                                  \aMatchingTuple \isnot \None}
                                    {\label{new-ttj-v2:amatchingtuple-when-matchingtuples-non-nil}
                                      \Return
                                      \Row $\concat$ \aMatchingTuple
					\label{new-ttj-v2:return-join-result-when-matchingtuples-non-nil}
				}
				\Row \Equal \Router.\FGetNext{}
				\label{tt-1-beta-b-hashtable-tq:get-next-get-r1}
				
				\lIf{\Row \is \None} {
					\Return \None \label{tt-1-beta-b-hashtable-tq:return-nil-in-matchingtuples-block}
				}
			}
			
			\lIf {\Row \is \None \myverb{||}
                          $\matchingtuples$ \Equal \EmptyList} {
				\Row \Equal \Router.\FGetNext{}
				\label{tt-1-beta-b-hashtable-tq:get-next-get-r1-when-router-nil}
			}
			
			\While {\Row \isnot \None} { 
				\label{tt-1-beta-b-hashtable-tq:loop-start}
				
				\tcp{Find tuples from \Rinner joinable with \Row}
				$\matchingtuples$ \Equal \myverb{inner}.\FGet{\myverb{keys(plan[1..num(this)],Rinner)}}\label{tt-1-beta-b-hashtable-tq:initialize-l-lookUpH}

				\If {$\matchingtuples$ \isnot \None} {
					\aMatchingTuple \Equal $\matchingtuples$.\FNext{}
					\label{tt-1-beta-b-hashtable-tq:set-amatchingtuple}
					
					\Return \Row $\concat$ \aMatchingTuple
					\label{tt-1-beta-b-hashtable-tq:return-join-result}
				} \Else { \tcp{Backjump to
                                    \myverb{parent(Rinner)}}
                                  \Row \Equal \Router.\FPassContext{\myverb{parent(Rinner)}}\label{tt-1-beta-b-hashtable-tq:pass-context-get-r1}
                        } }
			
			\Return \None \label{tt-1-beta-b-hashtable-tq:return-nil-at-the-end-get-next}
		}

		\Fn{\FPassContextPrefix{Relation $R$}} {
                  \If{\Rinner \myverb{==} $R$}
                     { \label{tt-1-beta-b-hashtable-tq:check-parent-child}
                       Remove \aMatchingTuple from $\matchingtuples$
                       and
                       \myverb{inner} \label{tt-1-beta-b-hashtable-tq:remove-tuple-in-l-pass-context}
                     } \Else
                     {\label{tt-1-beta-b-hashtable-tq:check-parent-child-else}
                       \tcp{Has not reached $R$; backjumping
                         continues}
                       $\matchingtuples$ \Equal \None
				\label{tt-1-beta-b-hashtable-tq:matchingtuples-nil-pass-context}
				
				\Row \Equal \Router.\FPassContext{$R$} \label{tt-1-beta-b-hashtable-tq:pass-context-pass-context}
				
				\lIf{\Row \is \None} {
					\Return \None \label{tt-1-beta-b-hashtable-tq:return-nil-pass-context}
				}
			}
			
			\Return \myverb{this}.\FGetNext{}
			\label{new-ttj-v2:get-next-in-pass-context}
		}
	}
\end{algorithm}

\subsection{Preliminaries}
\label{sec:prelim}

A $\TTJ$ \emph{left-deep plan} is a binary tree where internal nodes
are $\TTJ$ iterators (\cref{algo:ttj-join}), and every right child
\Rinner and the left child \Router of the leftmost iterator are
leaves; all other plans are \emph{bushy}. Leaves serve as both nodes
and relations, indexed left-to-right as $\Leaf_1, \dots,
\Leaf_{|\Query|}$, with $|\Query|$ denoting the number of relations
and \verb|plan[1..i]| denoting the prefix sequence $(\Leaf_1, \dots,
\Leaf_i)$. Function \verb|keys(plan[1..i],R)| computes \verb|R|'s
\emph{key}: the attributes it shares with $\Leaf_1, \dots,
\Leaf_{i-1}$. We define \verb|inner| as the hash table built on
\Rinner, and $\concat$ as schema-resolving tuple concatenation. For an
iterator with left child $\Leaf_i$, its label is $i$. Finally,
\myverb{this} self-references the iterator, and \myverb{parent(R)}
identifies \verb|R|'s parent in the join tree. We use $\join$ to
denote a $\TTJ$ iterator instance.

\subsection{A Motivating Example}
\label{sec:motivating-example}

Consider the ACQ $\Query$ defined in \cref{sec:logical-physical}
evaluated over the database instance in \cref{fig:join-trees}, where
$x_i, w_i, z_1$ ($i \in \{1,2\}$) denote distinct constants. $\Query$
is acyclic if it admits a join tree. A \emph{join tree} is a tree
where there is a one-to-one correspondence between nodes in the tree
and relations in the query; two nodes are adjacent if and only if
their corresponding relations share common attribute(s); and the tree
must satisfy the \emph{running intersection property} ($\RIP$): for
any attribute $a$ that appears in $\Query$, all relations containing
$a$ form a connected subtree in the join tree. Once rooted, the
directed join tree dictates that execution must backjump along the
tree edges, rather than the query plan, upon a probe failure. Multiple
join trees can exist for a given query; for example,
\cref{fig:join-trees}~(a) and (c) illustrate two valid join trees for
$\Query$, differing in their parent mappings where \verb|parent(T)| is
$S$ in tree (a), but $R$ in tree (c).

Consider a defective $\TTJ$ iterator missing
Line~\ref{tt-1-beta-b-hashtable-tq:matchingtuples-nil-pass-context},
which is responsible for clearing the local state $\matchingtuples$
during a backjump. Testing this defective iterator under join tree (a)
falsely reports success, while tree (c) correctly reports output
inconsistency, exposing the flaw.

Suppose we evaluate $\Query$ under join tree (a) using the plan from
\cref{fig:join-trees}~(b). Following the pipelined execution, the
iterator $\join_1$ fetches $R(13, x_1)$, successfully probes $S$, and
caches the match $S(13,w_1)$ in $\matchingtuples$ (\circleone). When
the subsequent probe against $T$ fails (\circletwo),
Line~\ref{tt-1-beta-b-hashtable-tq:pass-context-get-r1} from $\join_2$
is invoked, and the execution backjumps to $\join_1$ due to its right
child $S$ being the parent of $T$. This correctly invokes
Line~\ref{tt-1-beta-b-hashtable-tq:remove-tuple-in-l-pass-context},
flushing the stale state in $\matchingtuples$ (\circlethree). The
iterator $\join_1$ then fetches $R(14, x_2)$, probes both hash tables
successfully (\circlefour), and yields the correct output.

Conversely, evaluating $\Query$ under join tree (c) using the plan
from \cref{fig:join-trees}~(d) exposes the defect. Upon the same probe
failure at $T$ (\circlefive), iterator $\join_2$ calls
$\join_1$.\verb|deleteDT(R)| via
Line~\ref{tt-1-beta-b-hashtable-tq:pass-context-get-r1}. Because the
defective iterator lacks
Line~\ref{tt-1-beta-b-hashtable-tq:matchingtuples-nil-pass-context},
it skips clearing the state for the bypassed relation $S$. Instead,
$\join_1$ directly calls \verb|R.deleteDT(R)| (which simply invokes
\verb|getNext()| of $R$) from
Line~\ref{tt-1-beta-b-hashtable-tq:pass-context-pass-context}
(\circlesix). Consequently, $\matchingtuples$ incorrectly retains the
stale tuple $S(13, w_1)$. When $R(14, x_2)$ is fetched, the iterator
$\join_1$ processes this corrupted state, fails its validity check at
Line~\ref{new-ttj-v2:amatchingtuple-when-matchingtuples-non-nil}
(since $S(13,w_1)$ has already been processed), and erroneously
terminates with an empty result via
Line~\ref{tt-1-beta-b-hashtable-tq:return-nil-in-matchingtuples-block}.

Comparing these traces reveals a critical pitfall: a flawed physical
translation can easily evade detection and produce correct results if
the testing benchmark happens to utilize a forgiving topology like the
degenerate tree (a).  Relying solely on macro-benchmark queries leaves
the detection of such state-corruption defects to chance. This
demonstrates the absolute necessity of structure-aware testing by
manipulating join tree topologies directly during the early
translation stage to isolate exact root causes.

\section{$\SAAD$ Framework}
\label{sec:framework}

\begin{figure}[ht]
	\centering
	\includegraphics[width=\columnwidth]{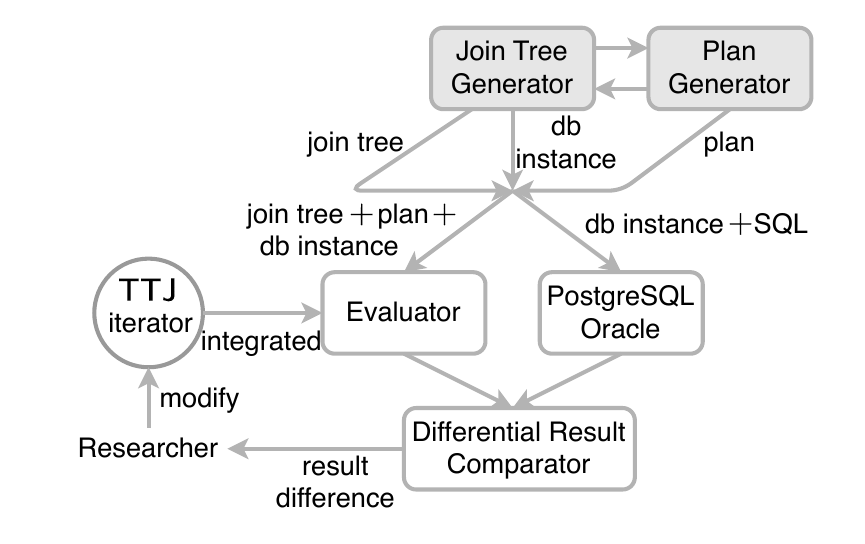}
	\caption{Overview of the $\SAAD$ framework used during the
          design of the $\TTJ$ iterator (\cref{algo:ttj-join}). Test
          case synthesis modules are gray rounded rectangles,
          while differential evaluation modules are white rounded
          rectangles. The $\TTJ$ iterator implementation is depicted
          as a~circle.}
	\label{fig:method}
\end{figure}

\begin{table*}[t]
  \caption{Diagnostic utility comparison between a standard
    macro-benchmark (TPC-H) and the $\SAAD$ framework across critical
    logical-to-physical translation defects. Note: TPC-H
      counts represent queries, whereas $\SAAD$ counts represent
      MREs.}
  \label{tab:bug-detection}
  \centering
  \begin{tabular}{p{4.8cm} p{2cm} p{2cm} p{7.5cm}}
    \toprule \textbf{Defect Class} & \textbf{TPC-H} & \textbf{$\SAAD$
      Suite} & \textbf{Limitation of Macro-Benchmarks} \\ \midrule
    \textbf{1. State Mismanagement} & Detected \newline (4 queries) &
    Isolated \newline (8 MREs) & Requires a specific combination of
    non-degenerate join tree topologies and adversarial data
    distributions to trigger. \\ \midrule \textbf{2. Reverse GYO
      Reduction Order Violation} & Missed \newline (0 queries) &
    Isolated \newline (1 MRE) & Macro-benchmark query plans are
    structurally homogeneous (easily mapped to join trees), failing to
    stress-test $\TTJ$'s theoretical precondition. \\ \midrule
    \textbf{3. $\RIP$ Violation} & Detected \newline (2 queries) &
    Isolated \newline (1 MRE) & Complex bushy plans from
    macro-benchmark queries introduce significant structural noise,
    severely hindering precise defect isolation. \\ \bottomrule
  \end{tabular}
\end{table*}

\begin{figure}
\begin{lstlisting}[escapeinside={(*}{*)}]
def generateRandomTree(n):
  nodes = [Node() for _ in range(n)]
  T = Tree()
  root = random_choice(nodes)
  S = [root]; nodes.remove(root)
  while nodes:                 
    u = random_choice(nodes)
    v = random_choice(S)     
    T.addEdge(v, u); S.append(u); nodes.remove(u)
  return (T, root)

def generate(max_size, max_rel_size, attr_pool):
  T, root = generateRandomTree(random(1, max_size))
  inherited = {root: attr_pool}
  x = random(1, len(attr_pool))
  queue = [root]
  while queue:
    next_level = []
    for node in queue:
      # sample x attrs from inherited pool, enforcing (*$\RIP$*)
      schema = sample_nonempty(inherited[node], x)  
      tuples = gen_tuples(schema, random(1, max_rel_size))
      postgres.ingest(Relation(node.name, schema), tuples)
      for child in T.children(node):
        # child inherits parent schema 
        inherited[child] = schema   (*\label{lst:inherit-child}*)
        next_level.append(child)
    # narrow pool after each BFS level
    x = max(x - random(0, x), 1)    
    queue = next_level
  return T
\end{lstlisting}
\caption{Synthesizing a join tree and a database instance in $\SAAD$ in
  the absence of an input query plan.}
\label{fig:generator}
\end{figure}

\cref{fig:method} illustrates the architecture of $\SAAD$ utilized
during the design of the $\TTJ$ iterator. While implemented for
$\TTJ$, $\SAAD$ is highly modular: its algorithm-agnostic components
readily support other join-tree algorithms by simply replacing the
$\TTJ$ iterator in the figure. $\SAAD$ operates in two tightly coupled
phases: test case synthesis and differential evaluation.

To isolate defects efficiently, $\SAAD$ synthesizes MREs rather than
massive benchmark workloads, capping query size at \MaxQSize = 5
relations and relation size at \MaxRelSize = 10 tuples. Notably, this
small configuration is sufficient to expose the defects we highlight
in \cref{sec:evaluation}. The synthesis pipeline consists of the
join tree generator and the plan generator, either of
which can be used to initiate the synthesis process. The plan
generator synthesizes a collection of relations and applies random
permutations and parenthesizations to construct arbitrary left-deep or
bushy query plans. \cref{fig:generator} illustrates the functionality
of the join tree generator when generating a query expressed directly
as a join tree. The \verb|generateRandomTree| function first builds a
random directed tree. Then, the \verb|generate| function enforces
$\RIP$ by utilizing a top-down breadth-first traversal: each child
node restricts its attribute pool to a random subset of its parent's
schema (Line~\ref{lst:inherit-child}). This structural inheritance
guarantees that the resulting tree inherently satisfies $\RIP$ and
hence is a valid join tree. Relation instances are populated alongside
the join tree construction and subsequently ingested into the
PostgreSQL oracle.

During evaluation, $\SAAD$ translates the ACQ's logical semantics into
a standard \verb|SELECT-FROM-WHERE| SQL query to obtain the ground
truth from PostgreSQL. Enforcing a specific join-tree topology is
unnecessary, as PostgreSQL's internal execution plan is strictly
orthogonal to the logical result. For exact output equivalence,
synthesized relations exclude \verb|NULL|s, and the SQL uses bag
semantics to match $\TTJ$'s multiset output. Simultaneously, $\SAAD$
binds $\TTJ$ iterators to the query plan. A differential comparator
then cross-verifies outputs to flag isolated defects.

Initially, $\SAAD$ fed plans directly to the evaluator, forcing $\TTJ$
to map the plan to a default join tree at runtime. As demonstrated in
\cref{sec:case-study-2,sec:case-study-3}, this direct coupling led to
semantic collapses. This forced $\SAAD$ into its current architecture
(\cref{fig:method}), where the join tree generator acts as a mandatory
validation layer (code omitted for brevity). Furthermore, when the
generator is tasked with constructing a valid join tree from an
arbitrary bushy plan, the structural translation becomes highly
non-trivial. We defer the detailed explanation of this complex
procedure to \cref{sec:case-study-2,sec:case-study-3}, where we
contextualize it alongside the specific logical-to-physical
translation defects it resolved.

\section{Evaluation}
\label{sec:evaluation}

In this section, we demonstrate the critical role of $\SAAD$ in the
physical translation of $\TTJ$ through three postmortem case
studies. These defects fundamentally shaped the theoretical boundaries
of $\TTJ$, its physical iterator implementation, and the architectural
evolution of $\SAAD$ itself. The $\SAAD$ suite is released along with
the source code of the $\TTJ$ query engine (implemented in Java);
comprehensive reproducibility instructions, including hardware
specifications and the PostgreSQL version (v13), can be found in our
previous work~\cite{Hu2025}. To illustrate the limitations of relying
on macro-benchmarks as a testing proxy, and to demonstrate how $\SAAD$
systematically addresses these blind spots, we benchmark our findings
against the same 13 standard TPC-H multi-join queries~\cite{TPC} (Q3,
Q7--Q12, Q14--Q16, Q18--Q20) as we did in~\cite{Hu2025}. Given the
retrospective nature of this paper, we defer a more rigorous
comparison with sophisticated database testing frameworks to future
work.

\cref{tab:bug-detection} summarizes the comparison. TPC-H counts
represent complex, end-to-end benchmark queries, whereas $\SAAD$
isolates defects using MREs. Overall, $\SAAD$ synthesized 28 failing
cases across the entire development lifecycle of $\TTJ$. For this
evaluation, we distilled these into 10 distinct MREs (categorized into
three representative defect classes in \cref{tab:bug-detection}) by
filtering out transient defects associated with early prototypes and
grouping variants of core translation errors by their underlying root
causes. The following subsections detail three representative
MREs. The first defect (\cref{sec:case-study-1}) demonstrates that
exposing intricate state corruptions requires a precise combination of
join tree topology and database instance. While TPC-H occasionally
triggers the defect, there is no clean topological categorization
between queries that mask the issue and those that expose it. This
ambiguity forces researchers to dive into the massive, noisy database
instance to isolate the root cause. In contrast, $\SAAD$ pinpoints the
exact failure using merely three relations and five tuples. The second
defect (\cref{sec:case-study-2}) exposes a deeply hidden theoretical
limitation of $\TTJ$ that TPC-H entirely missed, which ultimately
forced us to refine the formal precondition of $\TTJ$. The third
defect (\cref{sec:case-study-3}) highlights a structural translation
hazard. While TPC-H eventually surfaced the defect using sprawling 6-
or 8-relation bushy plans, $\SAAD$ achieved topological isolation by
reproducing the identical defect using a minimal structure of just
three relations.

\subsection{Case Study 1: State Mismanagement}
\label{sec:case-study-1}

This case study revisits the state mismanagement introduced in
\cref{sec:motivating-example}, caused by the omission of
Line~\ref{tt-1-beta-b-hashtable-tq:matchingtuples-nil-pass-context}
from \cref{algo:ttj-join}. While TPC-H benchmark queries eventually
surface the defect, they offer poor diagnostic utility. The defect was
flagged by only four queries (Q3, Q7, Q8, and Q10), all of which
possess non-degenerate join trees (i.e., at least one internal node
has multiple children, akin to \cref{fig:join-trees}~(c)). The defect
remained completely dormant in the remaining nine
queries. Unsurprisingly, the defective $\TTJ$ iterator masked the
defect across the seven queries featuring degenerate join trees (akin
to \cref{fig:join-trees}~(a)). However, the defect also remained
hidden in Q9 and Q11 despite their non-degenerate topologies. As a
result, there is no clean topological categorization between the TPC-H
queries that mask the defect and those that expose it.

The above observation suggests that topological structure alone is
insufficient to trigger the defect; it requires specific data
distributions. The inconsistent coverage of TPC-H benchmark queries
forces researchers to sift through the massive, noisy database
instance hoping for coincidental data alignment. By contrast, the
value of $\SAAD$ lies not just in detection, but in immediate
isolation: it deterministically synthesizes a minimal MRE consisting
of just three relations and five tuples
(\cref{sec:motivating-example}), bypassing the overhead of
macro-benchmark noise.

\subsection{Case Study 2: Refining the Theoretical Precondition}
\label{sec:case-study-2}

The next two case studies concern the physical translation of the
\verb|parent(Rinner)| logic
(Line~\ref{tt-1-beta-b-hashtable-tq:pass-context-get-r1} of
\cref{algo:ttj-join}). To achieve its theoretical guarantees, the
iterator must dynamically recover a valid join tree from the provided
plan. Let $\Attr{R}$ denote the attribute set of relation $R$. The
recovery procedure computes the key $K$ of \Rinner and assigns
\Rinner's parent to $\Leaf_j$, where $j$ is $\min\SetOf{q \in
  \Set{1,\dots,i-1}}{K \subseteq \Attr{\Leaf_q}}$. This produces the
shallowest possible join tree, maximizing backjumping distance on
probe failures and pruning redundant intermediate results.

Initially, $\SAAD$ sent left-deep plans directly to the evaluator
under the naive assumption that mapping any ACQ left-deep plan to a
join tree was trivial and could thus be performed during the
evaluation phase. This assumption catastrophically failed when $\SAAD$
synthesized a simple left-deep plan $P = R(a,b) \Join S(b,c) \Join
T(a,b,c)$, instantly triggering a \verb|NullPointerException| at
Line~\ref{tt-1-beta-b-hashtable-tq:pass-context-get-r1}. While $P$ is
undeniably an ACQ (admitting a valid join tree rooted at $T$), the
recovery failed because the key of $T$ is $\Set{a,b,c}$, yet neither
$\Attr{R}$ nor $\Attr{S}$ contains this set. This defect directly
refined our initial theoretical conjectures about
$\TTJ$.\footnote{Remarkably, this defect was concurrently identified
by a referee for our ICDT '25 submission.} We originally assumed
$\TTJ$ guaranteed linear-time evaluation on arbitrary left-deep
plans. Confronted with $\SAAD$'s counterexample (and corroborated by
formal review), we restricted $\TTJ$'s optimality
guarantee~\cite{Hu2025} strictly to left-deep plans satisfying the
reverse of a GYO reduction order.

Relying solely on macro-benchmarks would have created the illusion
that optimality applies to arbitrary left-deep plans, leaving this
theoretical flaw permanently undiscovered. A follow-up
study~\cite{Luo2026} proved that across 9,285 acyclic queries in five
popular macro-benchmarks (including TPC-H), every connected left-deep
plan satisfies the reverse of a GYO reduction order. This structural
homogeneity inherently masks the algorithm's true theoretical
boundary.

\subsection{Case Study 3: Architectural Evolution and Topological Isolation}
\label{sec:case-study-3}

Following the discovery of the GYO reduction order constraint
(\cref{sec:case-study-2}), we initially attempted a lightweight patch
by adding an inline validation check within
\verb|parent(Rinner)|. However, this check consistently raised false
alarms when $\SAAD$ injected valid ACQ bushy plans into the evaluator.

To understand the topological mapping failure, consider two bushy
plans derived from the same relations: $P_1 = R(a,b) \Join (T(a,b,c)
\Join S(b,c))$ and $P_2 = T(a,b,c) \Join (R(a,b) \Join S(b,c))$. For
$P_1$, recursively applying the runtime recovery method from
\cref{sec:case-study-2} operates bottom-up: it first builds a subtree
$T \rightarrow S$ for the inner subplan $T \Join S$, and then connects
the outer relation $R$ to the subtree root $T$, correctly yielding the
degenerate join tree $R \rightarrow T \rightarrow S$. However,
applying the identical logic to $P_2$ first builds a subtree $R
\rightarrow S$ for $R \Join S$, and then connects the outer relation
$T$ to the root $R$, yielding a tree $T \rightarrow R \rightarrow
S$. This topology is not a join tree because it violates $\RIP$: $T$
and $S$ share attribute $c$, but are disconnected by $R$.

To maintain a unified execution logic of the evaluator that maps any
query plan to a single valid join tree, we introduce the abstraction
of \emph{virtual relations} to restore $\RIP$ in the constructed tree
for a bushy plan like $P_2$. Instead of directly connecting the outer
relation $T$ with $R$, we abstract the $R \Join S$ subplan as a
virtual relation $V(a,b,c)$ where its attribute set $\Attr{V}$ is the
union of $\Attr{R}$ and $\Attr{S}$. The resulting valid join tree
roots at $T(a,b,c)$, connected to $V(a,b,c)$, which in turn has two
children $R(a,b)$ and $S(b,c)$. This virtual relation abstraction
successfully accommodates all 139 bushy plans derived from our
macro-benchmark study on $\TTJ$~\cite{Hu2025}. (A formal
characterization of when this approach is applicable is established by
Lemma~\ref{lem:bushy} in Appendix~\ref{sec:bushy-plan-convertible}.)

Building on the breakdown of the trivial runtime mapping assumption
(\cref{sec:case-study-2}), the complexity of virtual relations proved
that in-iterator join tree recovery is architecturally
unsustainable. This realization triggered the architectural evolution
of $\SAAD$ to its final state (\cref{fig:method}), entirely offloading
the responsibilities of GYO validation (\cref{sec:case-study-2}) and
join tree construction to $\SAAD$'s synthesis pipeline. While TPC-H Q8
(an 8-relation join) and Q9 (a 6-relation join) eventually triggered
this $\RIP$ violation, their sprawling plans were unnecessarily
complicated, requiring significant effort to trace the structural root
cause. By contrast, $\SAAD$ reproduced the identical defect using a
minimal structure of just three relations.

\section{Related Work}
\label{sec:related-work}

Extensive research applies differential and generator-based testing to
validate database
implementations~\cite{Rigger2020,Seltenreich,Jung2019,Tang2023,Song2023,Jiang2024,Zhong2020},
typically executing synthesized SQL queries against mature
DBMSs. While sharing this differential testing spirit, $\SAAD$
differentiates itself in three ways. (1) Agility for nascent
algorithms: unlike existing tools that assume stable architectures,
$\SAAD$ is lightweight, avoiding the prohibitive overhead of
integrating highly volatile algorithm prototypes (e.g.,~$\TTJ$) into
mature systems. (2) Micro-topological synthesis: instead of operating
at the macroscopic SQL interface and relying on query optimizers whose
heuristics may miss problematic plans and join
trees~\cite{Zhao2025,Hu2025,Luo2026}, $\SAAD$ directly synthesizes
join trees and query plans, allowing it to stress-test both physical
implementations and theoretical preconditions. (3) Topological
isolation: by focusing on join trees, $\SAAD$ implicitly categorizes
multi-join queries based on acyclicity, enabling
minimal-by-construction test cases for precise defect isolation. This
focus on minimality reflects a pragmatic design choice: $\SAAD$ does
not aim to completely eliminate logical-to-physical translation
defects; it simply aims to serve as a first-line filter for those that
are structure-dependent. Hence, unlike earlier work such as
QuickCheck~\cite{Claessen2000}, which generates increasingly larger
test cases, $\SAAD$ heuristically imposes an upper bound on their
size; any defect requiring a test case beyond this bound is left to be
discovered during the formal proof stage.

\section{Concluding Remarks}
\label{sec:conclusion}

Our retrospective analysis of $\SAAD$ demonstrates that bridging
logical specifications and physical implementations is a bidirectional
feedback loop: early structural testing hardened $\TTJ$'s complex
iterator and directly refined its theoretical precondition. This
highlights the danger of relying solely on macro-benchmarks, which
actively mask translation defects. As a valuable early-stage vanguard
to formal verification, $\SAAD$ isolates non-trivial defects via
minimal reproducible examples long before the prohibitive effort of
writing formal proofs. Future extensions include integrating
state-of-the-art join tree enumeration~\cite{Luo2026} to optimize test
construction, bridging $\SAAD$ with mature database testing frameworks
to establish a structure-aware testing harness for other
join-tree-based algorithms in production-grade systems (e.g., adapting
$\SAAD$ for Yannakakis's algorithm~\cite{Yannakakis1981} simply
requires replacing the $\TTJ$ iterator, as core components like the
join tree generator remain algorithm-agnostic), and generalizing this
structure-aware methodology to the early conceptualization phase of
novel algorithms.

\begin{acks}
  The author thanks Daniel Miranker for early guidance on the $\TTJ$
  project, advisor Remy Wang for helpful discussions, and Yixuan Ni
  and Zhiting Zhu for manuscript feedback.
\end{acks}

\bibliographystyle{ACM-Reference-Format}
\bibliography{references}

\appendix

\section{The Existence of a Join Tree for a Bushy Plan of an ACQ}
\label{sec:bushy-plan-convertible}

Since the virtual relation method is not exclusive to $\TTJ$, we first
generalize the $\TTJ$ left-deep plan from \cref{sec:prelim}, allowing
its internal nodes to represent arbitrary join implementations.

\begin{definition}
  A left-deep plan is a binary tree where internal nodes are joins,
  and every right child, along with the left child of the leftmost
  join, is a leaf.
\end{definition}

We assume that for a given plan of a query $\Query$, the leaves of the
plan are labeled from left to right in increasing order $\Leaf_1,
\dots, \Leaf_{|Q|}$. We iterate through the leaves from right to left
and find the first \emph{maximal} left-deep subplan. A left-deep
subplan is maximal if adding any additional internal node and its
associated children from the bushy plan prevents the resulting plan
from being left-deep. We call a plan $P$ \emph{terminal} if the first
maximal left-deep subplan of $P$ is $P$ itself. Thus any left-deep
plan is terminal.

For any left-deep subplan $q$ of $P$, we consider $q$ to be a query
and hence $|q|$ is the number of relations in $q$. Let $\Attr{R}$
produce the attribute set of relation $R$. Abusing a bit of notation,
let $\Attr{P}$ denote the union of attributes of relations that are
leaves of plan $P$ and let $\Attr{T}$ denote the union of attributes
of all relations associated with nodes in a join tree $T$. Let
$\Par{R}{P}$ take in a relation $R$ and a plan $P$ and return the
parent of the node associated with $R$ in plan $P$; let $\Child{R}{P}$
take in a relation $R$ and a plan $P$ and return the right child of
the node associated with $R$ in plan $P$; and let $\Root{P}$ return
the root node of a plan $P$. When we label a leaf with a plan name as
the superscript, the subscript of the leaf indicates the rank of the
leaf within the plan indicated by the superscript name.

\begin{definition}
  Let $T_1 = (X_1, E_1)$ (resp., $T_2 = (X_2, E_2)$) be a tree where
  $X_1$ (resp., $X_2$) is a set of nodes and $E_1$ (resp., $E_2$) is a
  subset of $X_1 \times X_1$ (resp., $X_2 \times X_2$). Moreover, $X_1
  \cap X_2 = \{u\}$ for some node $u$. We define $\Merge{T_1}{T_2}$
  for $T_1$ and $T_2$ as follows. The output of $\Merge{T_1}{T_2}$ is
  a graph $T = (X, E)$ where $X = X_1 \cup X_2$ and $E = E_1 \cup
  E_2$. It is straightforward to see that $T$ is a tree.
\end{definition}

\begin{definition}[the reverse of a GYO reduction order~\cite{Hu2025}]
  \label{def:reverse-gyo}
  Let $P$ be a left-deep plan with leaves $\Leaf_1, \dots,
  \Leaf_{|Q|}$. Plan $P$ satisfies the reverse of a GYO reduction
  order if for every $\Leaf_j$ ($2 \le j \le |Q|$), there exists some
  $\Leaf_i$ ($1 \le i < j$) such that $K \subseteq \Attr{\Leaf_i}$
  where $K$ is equal to $\Attr{\Leaf_j} \cap \big(\bigcup_{r=1}^{j-1}
  \Attr{\Leaf_r})$.
\end{definition}

\begin{definition}[replacement]
  \label{def:replacement}
  We define a sequence of replacements applied towards a plan $P$ as
  follows. Let $P^i$ denote the input to the $i$th replacement ($i \ge
  0$). Let $P^0$ be $P$. Replacement $i$ takes in the plan output of
  replacement $i-1$ and outputs another plan. (Replacement $0$ takes
  in $P^0$.) Consider the $i$th replacement $\Repls{P^i}$. Let $q_i$
  be the first maximal left-deep subplan of $P^i$. Let $V$ denote the
  virtual relation of $q_i$ where $\Attr{V}$ is $\Attr{q_i}$. The
  output plan $P^{i+1}$ is constructed as follows. We set
  $\Child{\Par{\Root{q_i}}{P}}{P}$ to be $V$. Then we discard all the
  nodes that appear in both $q_i$ and $P$ from $P$. The remaining $P$
  is $P^{i+1}$. The replacement sequence stops after $\Repls{P^m}$
  (for some $m > 0$) where $P^m$ is terminal.
\end{definition}

From Definition~\ref{def:replacement}, we have the following
observation.

\begin{proposition}
\label{prop:obs}
Each replacement produces one virtual relation and the virtual
relation is a leaf node that is not the leftmost.
\end{proposition}

Hence let $\Virt{P^{i+1}}$ denote the lone virtual relation that is
built from $\Repls{P^i}$.

\begin{definition}[nice]
  \label{def:nice}
  A bushy plan $P$ is \emph{nice} if (1) $P$ does not have Cartesian
  products, (2) any left-deep subplan $q$ of $P$ satisfies the reverse
  of a GYO reduction order, and (3) for all $i \ge 0$, let $\Leaf_j$
  denote $\Virt{P^{i+1}}$ for some $j > 1$, there exists $1 \le r < j$
  such that $K \subseteq \Attr{\Leaf_r}$ where $K$ is the key of
  $\Leaf_j$.
\end{definition}

\begin{lemma}
  \label{lem:replacement}
  Suppose a sequence of $m$ replacements applied to a nice bushy plan $P$:
  $\Repls{P^0}, \dots, \Repls{P^m}$. Then for all $0 \le i \le m$, the
  first maximal left-deep subplan $q_i$ of $P^i$ satisfies the reverse
  of a GYO reduction order.
\end{lemma}
\begin{proof}
  Let $0 \le r \le m$ and let $G(r)$ denote the predicate ``the first maximal
  left-deep subplan of $P^r$ satisfies the reverse of a GYO reduction
  order.'' We prove by induction that $G(r)$ holds for all $r =
  0,\dots,m$.

  Base case: By the left-deep plan definition and
  Definition~\ref{def:nice}~(2), the first maximal left-deep subplan
  $q_0$ of $P^0$ satisfies the reverse of a GYO reduction order.

  Induction step: Let $r$ belong to $\{1,\dots,m\}$ and assume that
  $G(r')$ holds for all $r'$ in $\{0,\dots,r-1\}$. We need to prove
  that $G(r)$ holds. Let $P^r$ be the output of
  $\Repls{P^{r-1}}$. Consider the first maximal left-deep subplan
  $q_r$ of $P^r$ and a leaf $\Leaf_i^{q_r}$. Since $i=1$ holds
  vacuously under Definition~\ref{def:reverse-gyo}, we consider two
  cases where $i > 1$.

  Case~1: $\Leaf_i^{q_r}$ is a non-virtual relation. Leaf node
  $\Leaf_i^{q_r}$ is non-virtual in $P$. Suppose $\Leaf_i^{q_r}$
  corresponds to $\Leaf_{i'}^{P}$ in $P$. By
  Definition~\ref{def:nice}~(2) and Definition~\ref{def:reverse-gyo},
  there exists a leaf $\Leaf_{j'}^P$ for some $1 \le j' < i'$ such
  that the key of $\Leaf_{i'}^P$ is a subset of $\Attr{\Leaf_{j'}^P}$
  and $\Leaf_{i'}^P$ belongs to the same left-deep subplan as
  $\Leaf_{j'}^P$. Since $\Leaf_{i'}^{P}$ is not in $q_{r-1}$, and $1
  \le j' < i'$, $\Leaf_{j'}^P$ is not in $q_{r-1}$. Furthermore, since
  all the nodes in the first maximal left-deep subplan of $P^{r'}$ for
  all $r'$ in $\{0,\dots, r-1\}$ are discarded per
  Definition~\ref{def:replacement}, $\Leaf_{j'}^P$ is not in
  $q_{r'}$. Thus $\Leaf_{j'}^P$ corresponds to $\Leaf_j^{q_r}$ for
  some $1 \le j < i$. Hence $\Attr{\Leaf_j^{q_r}} =
  \Attr{\Leaf_{j'}^P}$, which contains the key of $\Leaf_i^{q_r}$.

  Case~2: $\Leaf_i^{q_r}$ is a virtual relation. By
  Proposition~\ref{prop:obs}, we have $i > 1$. By
  Definition~\ref{def:nice}~(3), there exists a preceding relation
  $\Leaf_j^{q_r}$ ($1 \le j < i$) whose attributes contain the key of
  $\Leaf_i^{q_r}$.

  Combining the two cases, we conclude that $G(r)$ holds, as required.
\end{proof}

\begin{lemma}
  \label{lem:bushy-helper}
  If a left-deep plan $P$ satisfies the reverse of a GYO reduction
  order, then $P$ corresponds to a join tree with root $\Leaf_1^P$.
\end{lemma}
\begin{proof}
  Immediate from Lemma 4.6 of \cite{Hu2025}.
\end{proof}

\begin{lemma}
  \label{lem:virtual-gyo}
  Let $P$ be a left-deep plan with leaves $\Leaf_1, \dots,
  \Leaf_{|P|}$ and let $V$ be a virtual relation where $\Attr{V} =
  \Attr{P}$. Then the sequence $S = \{V, \Leaf_1, \dots,
  \Leaf_{|P|}\}$ satisfies the reverse of a GYO reduction order.
\end{lemma}
\begin{proof}
  Since $\Attr{V}$ is $\Attr{P}$, the key of $\Leaf_1$ is a subset of
  $\Attr{V}$. Let $1 < i \le |P|$. Since $\Attr{V}$ is $\Attr{P}$, the
  key of $\Leaf_i$ in $S$ is $\Attr{\Leaf_i} \cap \big(\Attr{V} \cup
  \bigcup_{r=1}^{i-1} \Attr{\Leaf_r}) = \Attr{\Leaf_i} \cap \Attr{V}$,
  which is contained in $\Attr{V}$. Since $i$ is chosen arbitrarily,
  we conclude that $S$ satisfies the reverse of a GYO reduction order,
  as required.
\end{proof}

\begin{definition}[virtual relation method]
  \label{def:virtual-relation-method}
  Let $P$ be a nice bushy plan. We enhance a sequence of replacements
  in Definition~\ref{def:replacement} to construct a tree as follows.
  Let $P^i$ denote the input to the $i$th replacement ($i \ge 0$). We
  extend $\Repls{P^i}$ with a tree $T_i$ as an additional input. The
  function produces a tree $T_{i+1}$ as an additional output. Thus we
  define a new function $\Repl{P^i}{T_i}$, which produces a tuple of
  two elements $P^{i+1}$ and $T_{i+1}$. Let $T_0$ be an empty set.

  Consider the $i$th replacement $\Repl{P^i}{T_i}$. Let $q_i$ be the
  first maximal left-deep subplan of $P^i$. Let $V$ denote the virtual
  relation of $q_i$ where $\Attr{V}$ is $\Attr{q_i}$. The output plan
  $P^{i+1}$ is constructed the same as
  Definition~\ref{def:replacement}. Tree $T_{i+1}$ is constructed as
  follows. Let $q_i$ be the first maximal left-deep subplan of
  $P^i$. We define a function $\Choose{q_i} = \Virt{P^{i+1}}$. By
  Lemmas~\ref{lem:bushy-helper} and~\ref{lem:virtual-gyo}, we
  construct a join tree $T^{q_i}$ with root $\Choose{q_i}$ via the
  method described in \cref{sec:case-study-2}. Then $T_{i+1}$ is
  $\Merge{T^{q_i}}{T_i}$.
\end{definition}

\begin{proposition}
  \label{prop:attrs}
  Let plan $P^i$ and tree $T_i$ denote the output of
  $\Repl{P^{i-1}}{T_{i-1}}$ for some $i \ge 1$. Let $q_{i-1}$ denote
  the first maximal left-deep subplan of $P^{i-1}$. Then
  $\Attr{\Virt{P^i}}$ is equal to $\Attr{q_{i-1}}$, which is equal to
  $\Attr{T_i}$.
\end{proposition}
\begin{proof}
  Immediate from Definition~\ref{def:virtual-relation-method}.
\end{proof}

We remark that the virtual relation method does not actually modify
$P$. Since the virtual relation method enhances
Definition~\ref{def:replacement} with tree construction, the
following lemma is immediate from Lemma~\ref{lem:replacement}.

\begin{lemma}
  \label{lem:replacement2}
  Suppose a sequence of $m$ replacements applied to a nice bushy plan $P$:
  $\Repl{P^0}{T_0}, \dots, \Repl{P^m}{T_m}$. Then for all $0 \le i \le m$, the
  first maximal left-deep subplan $q_i$ of $P^i$ satisfies the reverse
  of a GYO reduction order.
\end{lemma}

For a tree $T$ where each node corresponds to a relation, we identify
each node with its associated relation. For example, a relation is in
$T$ if there exists a node in $T$ that is associated with the
relation. We say an attribute $a$ appears in $T$ if there exists some
relation $R$ in $T$ such that $a \in \Attr{R}$. We say a relation $R$
contains attribute $a$ if $a \in \Attr{R}$.

\begin{lemma}
  \label{lem:merge}
  Let $T_1 = (X_1, E_1)$ (resp., $T_2 = (X_2, E_2)$) be a join tree
  where $X_1$ (resp., $X_2$) is a set of nodes and $E_1$ (resp.,
  $E_2$) is a subset of $X_1 \times X_1$ (resp., $X_2 \times
  X_2$). Moreover, $X_1 \cap X_2 = \{u\}$ for some node $u$ where the
  associated relation $R$ has $\Attr{R} = \Attr{T_2}$. Let $b$ be an
  attribute that appears in both $T_1$ and $T_2$. Relations in
  $\Merge{T_1}{T_2}$ that contain $b$ form a connected subtree.
\end{lemma}
\begin{proof}
  Let $T$ be $\Merge{T_1}{T_2}$. Since $T$ is a tree, it is sufficient
  to show that for any two relations in $T$, if the relations contain
  $b$, all the relations on the path between the two relations contain
  attribute $b$. Since $\Attr{R} = \Attr{T_2}$ and $b$ appears in
  $T_2$, $b$ is in $\Attr{R}$. Let $S$ be a relation in $T_1$ that
  contains $b$. Since $T_1$ is a join tree, $S$ and $R$ belong to a
  connected subtree, which indicates that all the relations on the
  path from $S$ to $R$ contain $b$. Let $W$ be a relation in $T_2$
  that contains $b$. By a similar argument as above, all the relations
  on the path between $R$ and $W$ contain $b$. Hence we conclude that
  all the relations on the path between $S$ and $W$ contain $b$. Since
  $S$ and $W$ are chosen arbitrarily, we conclude that the claim of
  the lemma holds.
\end{proof}

Lemma~\ref{lem:bushy} completes \cref{sec:case-study-3} and shows that
using the virtual relation method, it is possible to construct a join
tree from a nice bushy plan.

\begin{lemma}
  \label{lem:bushy}
  Let $P$ be a nice bushy plan of an ACQ. Then the virtual relation
  method from Definition~\ref{def:virtual-relation-method} constructs
  a join tree from $P$.
\end{lemma}
\begin{proof}  
  Suppose the sequence of replacements stops after the $m$th
  replacement (i.e., $\Repl{P^0}{T_0}, \dots, \Repl{P^m}{T_m}$). Let
  $1 \le r \le m+1$ and let $G(r)$ denote the predicate ``$T_r$ from
  $\Repl{P^{r-1}}{T_{r-1}}$ is a join tree.'' We prove by induction
  that $G(r)$ holds for all $r = 1,\dots,m+1$.

  Base case: By the definition of left-deep plan,
  Definition~\ref{def:nice}~(2), and Lemma~\ref{lem:bushy-helper},
  tree $T_1$ from the first maximal left-deep subplan $q_0$ of $P^0$
  is a join tree.

  Induction step: Let $r$ belong to $\{2,\dots,m+1\}$ and assume that
  $G(r')$ holds for all $r'$ in $\{1,\dots,r-1\}$. We need to prove
  that $G(r)$ holds. Let $q_{r-1}$ be the first maximal left-deep
  subplan of $P^{r-1}$. By Lemma~\ref{lem:replacement2}, $q_{r-1}$
  satisfies the reverse of a GYO reduction order. By
  Definition~\ref{def:virtual-relation-method}, tree $T^{q_{r-1}}$
  from $q_{r-1}$ is a join tree. Next we argue that $T_r$, the output
  of $\Merge{T^{q_{r-1}}}{T_{r-1}}$, is a join tree by showing it
  satisfies $\RIP$. By Proposition~\ref{prop:attrs}, we have
  $\Attr{\Virt{P^{r-1}}} = \Attr{T_{r-1}}$. Consider an attribute $a$
  that appears in some relation in $T_r$. If $a$ appears in $T_{r-1}$
  but not in $T^{q_{r-1}}$, or appears in $T^{q_{r-1}}$ but not in
  $T_{r-1}$, $\RIP$ holds trivially since $T_{r-1}$ and $T^{q_{r-1}}$
  are join trees. If $a$ appears in both $T_{r-1}$ and $T^{q_{r-1}}$,
  since $T^{q_{r-1}}$ and $T_{r-1}$ are join trees, $T^{q_{r-1}}$ and
  $T_{r-1}$ share exactly one node $\Virt{P^{r-1}}$, and
  $\Attr{\Virt{P^{r-1}}} = \Attr{T_{r-1}}$, by Lemma~\ref{lem:merge},
  we conclude that relations in $\Merge{T^{q_{r-1}}}{T_{r-1}}$ that
  contain $a$ form a connected subtree. Combining all three cases, we
  conclude that $G(r)$ holds, as required. Since $P^m$ is terminal and
  $G(r)$ holds, we conclude that the claim of the lemma holds.
\end{proof}

\end{document}